\documentclass[journal]{IEEEtran}
\usepackage[cmex10]{amsmath}
\usepackage{amssymb}

\newcommand{\bbC}{\mathbb{C}}
\newcommand{\bbF}{\mathbb{F}}
\newcommand{\bbR}{\mathbb{R}}
\newcommand{\bbZ}{\mathbb{Z}}

\newcommand{\bfB}{\boldsymbol{B}}
\newcommand{\bfc}{\boldsymbol{c}}
\newcommand{\bfe}{\boldsymbol{e}}
\newcommand{\bff}{\boldsymbol{f}}
\newcommand{\bfF}{\boldsymbol{F}}
\newcommand{\bfh}{\boldsymbol{h}}
\newcommand{\bfI}{\mathbf{I}}
\newcommand{\bfx}{\boldsymbol{x}}
\newcommand{\bfy}{\boldsymbol{y}}

\newcommand{\bfchi}{\boldsymbol{\chi}}
\newcommand{\bfone}{\boldsymbol{1}}
\newcommand{\bfphi}{\boldsymbol{\varphi}}
\newcommand{\bfPhi}{\boldsymbol{\Phi}}
\newcommand{\bfpsi}{\boldsymbol{\psi}}
\newcommand{\bfPsi}{\boldsymbol{\Psi}}

\newcommand{\calB}{\mathcal{B}}
\newcommand{\calD}{\mathcal{D}}
\newcommand{\calG}{\mathcal{G}}
\newcommand{\calH}{\mathcal{H}}
\newcommand{\calL}{\mathcal{L}}
\newcommand{\calM}{\mathcal{M}}
\newcommand{\calN}{\mathcal{N}}
\newcommand{\calR}{\mathcal{R}}
\newcommand{\calS}{\mathcal{S}}
\newcommand{\calV}{\mathcal{V}}

\newcommand{\rmi}{\mathrm{i}}
\newcommand{\rmT}{\mathrm{T}}

\newcommand{\hd}{\operatorname{d}}
\newcommand{\tr}{\operatorname{tr}}
\newcommand{\Tr}{\operatorname{Tr}}
\newcommand{\dist}{\operatorname{dist}}
\newcommand{\sgn}{\operatorname{sgn}}

\newcommand{\abs}[1]{|{#1}|}
\newcommand{\bigparen}[1]{\bigl({#1}\bigr)}
\newcommand{\biggparen}[1]{\biggl({#1}\biggr)}
\newcommand{\set}[1]{\{{#1}\}}
\newcommand{\norm}[1]{\|{#1}\|}
\newcommand{\ip}[2]{\langle{#1},{#2}\rangle}

\newtheorem{theorem}{Theorem}
\newtheorem{corollary}{Corollary}

\begin{document}
\title{Kirkman Equiangular Tight Frames and Codes}

\author{John~Jasper, Dustin~G.~Mixon and Matthew~Fickus~\IEEEmembership{Member,~IEEE}%
\thanks{J.~Jasper is with the Department of Mathematics, University of Missouri, Columbia, MO 65211, USA.}%
\thanks{D.~G.~Mixon and M.~Fickus are with the Department of Mathematics and Statistics, Air Force Institute of Technology, Wright-Patterson Air Force Base, OH 45433, USA, e-mail: Matthew.Fickus@gmail.com}}

\maketitle


\begin{abstract}
An equiangular tight frame (ETF) is a set of unit vectors in a Euclidean space whose coherence is as small as possible, equaling the Welch bound.
Also known as Welch-bound-equality sequences, such frames arise in various applications, such as waveform design and compressed sensing.
At the moment, there are only two known flexible methods for constructing ETFs: harmonic ETFs are formed by carefully extracting rows from a discrete Fourier transform; Steiner ETFs arise from a tensor-like combination of a combinatorial design and a regular simplex.
These two classes seem very different: the vectors in harmonic ETFs have constant amplitude, whereas Steiner ETFs are extremely sparse.
We show that they are actually intimately connected: a large class of Steiner ETFs can be unitarily transformed into constant-amplitude frames, dubbed Kirkman ETFs.
Moreover, we show that an important class of harmonic ETFs is a subset of an important class of Kirkman ETFs.
This connection informs the discussion of both types of frames: some Steiner ETFs can be transformed into constant-amplitude waveforms making them more useful in waveform design; some harmonic ETFs have low spark, making them less desirable for compressed sensing.
We conclude by showing that real-valued constant-amplitude ETFs are equivalent to binary codes that achieve the Grey-Rankin bound, and then construct such codes using Kirkman ETFs.
\end{abstract}


\section{Introduction}

An \textit{equiangular tight frame} (ETF) is a maximal packing of $N$ lines in an $M$-dimensional Euclidean space.
To be precise, the \textit{coherence} of $N$ unit vectors $\set{\bfphi_n}_{n\in\calN}$ in such a space is the maximal modulus of their inner products:
\begin{equation}
\label{equation.definition of coherence}
\mu:=\max_{n\neq n'}\abs{\ip{\bfphi_n}{\bfphi_{n'}}}.
\end{equation}
In applications such as waveform design~\cite{StrohmerH03} and compressed sensing~\cite{DeVore07}, one often wants vectors with low coherence.
Here, the best one can hope to achieve is the \textit{Welch bound}~\cite{Welch74}:
\begin{equation}
\label{equation.Welch bound}
\mu\geq\bigparen{\tfrac{N-M}{M(N-1)}}^\frac12.
\end{equation}
As detailed later on, equality in~\eqref{equation.Welch bound} is only achieved when the $\bfphi_n$'s are an ETF~\cite{StrohmerH03}, also known as a Welch bound equality (WBE) sequence.

ETFs are not easily found.
Indeed, despite over a decade of active research, only five general constructions of ETFs are known, and the first two of these are trivial: in any space of dimension $M$, we can always take an orthonormal basis ($N=M$) or a regular simplex ($N=M+1$).
All other known constructions involve combinatorial design.
For instance, one can build ETFs with $N=2M$ using conference matrices provided either $N=2^{j+1}$ or $N=p^j+1$ where $j$ is a positive integer and $p$ is an odd prime~\cite{StrohmerH03}.
This class of ETFs is limited in the sense that their \textit{redundancy} $N/M$ is necessarily two.
Moreover, they are closely related to other ETFs which are constructed using the two remaining methods~\cite{Renes07,Strohmer08}.

The fourth known construction method is much more versatile, relying on the well-studied topic of \textit{Abelian difference sets},
namely subsets $\calD$ of a commutative group $\calG$ with the property that the size of the set $\set{(d,d')\in\calD^2: g=d-d'}$ is independent of $g\in\calG$.
For decades, it has been known that a subset $\calD$ of $\calG$ is a difference set precisely when the modulus of the discrete Fourier transform (DFT) of its characteristic function is a perfect spike~\cite{Turyn65}.
As shown in \cite{StrohmerH03,XiaZG05,DingF07}, this implies that restricting the characters of $\calG$ to $\calD$ yields an ETF of $N=\abs{\calG}$ vectors for a space of dimension $M=\abs{\calD}$.
In particular, Singer and McFarland difference sets~\cite{JungnickelPS07} yield ETFs of almost arbitrary redundancy and size.

The fifth known construction method involves Steiner systems, namely a set $\calB$ of blocks (subsets) of a finite set $\calV$ which has the properties that (i) every block has the same number of points, (ii) every point is contained in the same number of blocks and (iii) any two points determine a unique block.
For many years, it has been known that such systems can be used to build \textit{strongly regular graphs}~\cite{GoethalsS70}.
Moreover, strongly regular graphs with certain parameters are known to be equivalent to real ETFs~\cite{HolmesP04,Waldron09}.
In~\cite{FickusMT12}, these ideas are distilled into a direct method for constructing real or complex ETFs via a tensor-like combination of the incidence matrix of Steiner system with a unimodular regular simplex.

Like harmonic ETFs, these \textit{Steiner ETFs} are extremely flexible, providing ETFs whose size and redundancy are arbitrary, up to an order of magnitude.
However, whereas harmonic ETFs have \textit{constant amplitude}---the entries of their frame vectors have constant modulus---Steiner ETFs are extremely sparse.
This sparsity can be a detriment in applications: in radio communication and radar, constant-amplitude waveforms allow more energy to be transmitted by power-limited hardware.
Moreover, though Steiner ETFs have optimal coherence, are thus good for coherence-based compressed sensing, they have terrible \textit{spark}: a small number of Steiner ETF elements can be linearly dependent~\cite{FickusMT12}.
As such, Steiner ETFs do not satisfy compressed sensing's Restricted Isometry Property (RIP) in a way that rivals that of random matrices.

In this paper, we provide a new method for unitarily transforming certain Steiner ETFs into constant-amplitude ETFs.
This method only works when the underlying Steiner system is \textit{resolvable}, meaning that its blocks $\calB$ can partitioned into several collections of blocks $\set{\calB_r}_{r\in\calR}$, where for any $r$, the blocks in $\calB_r$ form a partition of $\calV$.
Such systems were first made famous in 1850 by \textit{Kirkman's schoolgirl problem}, and as such, we dub these frames \textit{Kirkman ETFs}.

In the next section, we provide the basic mathematical background on Steiner ETFs.
In Section~3, we provide the Kirkman construction itself, and then use the existing literature on resolvable Steiner systems to construct several new families of constant-amplitude ETFs.
It turns out that one of these ``new" families---those arising from finite affine geometries---corresponds to one of the most important classes of harmonic ETFs, namely those constructed via McFarland difference sets;
as discussed in the fourth section, this identification allows us, for the first time, to seriously investigate the RIP properties of these McFarland ETFs.
In Section~5, we identify a real-valued constant-amplitude ETF with a self-complementary binary code, and in this context show that the Welch bound is equivalent to the Grey-Rankin bound of coding theory.
We then use the results from the previous sections to explicitly construct some Grey-Rankin-bound-equality codes.

\section{Preliminaries}
Throughout this paper, lowercase, uppercase and calligraphy denote an element of a finite set, the number of elements in that set, and the set itself, respectively.
For example, $m\in\calM$ where $M=\abs{\calM}$.
Also, $\bbC^{\calM}:=\set{\bfx:\calM\rightarrow\bbC}$ denotes the $M$-dimensional inner product space consisting of all complex-valued functions over $\calM$, and bold lowercase and uppercase denote vectors and operators/matrices, respectively.

The \textit{synthesis operator} of any vectors $\set{\bfphi_n}_{n\in\calN}$ in $\bbC^{\calM}$ is the linear operator $\bfPhi:\bbC^{\calN}\rightarrow\bbC^{\calM}$, $\bfPhi\bfy:=\sum_{n\in\calN}\bfy(n)\bfphi_n$.
In the special case $\calM=\set{1,\dotsc,M}$ and $\calN=\set{1,\dotsc,N}$, the synthesis operator is the $M\times N$ matrix $\bfPhi=[\bfphi_1\,\cdots\,\bfphi_N]$.
The \textit{analysis operator} $\bfPhi^*:\bbC^\calM\rightarrow\bbC^\calN$ is the adjoint of the synthesis operator, meaning $(\bfPhi^*\bfx)(n)=\bfphi_n^*\bfx=\ip{\bfx}{\bfphi_n}$.
In \textit{finite frame theory}, one seeks $\bfphi_n$'s that meet various application-motivated constraints and whose \textit{frame operator} $\bfPhi\bfPhi^*:\bbC^{\calM}\rightarrow\bbC^{\calM}$, $\bfPhi\bfPhi^*\bfx=\sum_{n\in\calN}\ip{\bfx}{\bfphi_n}\bfphi_n$ is as well-conditioned as possible.
In particular, $\set{\bfphi_n}_{n\in\calN}$ is a~\textit{unit norm tight frame} (UNTF) if $\norm{\bfphi_n}=1$ for all $n$ and if $\bfPhi\bfPhi^*=A\bfI$ for some constant $A$.
Here, since $MA=\Tr(A\bfI)=\Tr(\bfPhi\bfPhi^*)=\Tr(\bfPhi^*\bfPhi)=\sum_{n\in\calN}\norm{\bfphi_n}^2=N$, this constant $A$ is necessarily the frame's \textit{redundancy} $N/M$.

This paper is about \textit{equiangular tight frames} (ETFs), namely UNTFs $\set{\bfphi_n}_{n\in\calN}$ which have the additional property that the magnitude of $\ip{\bfphi_n}{\bfphi_{n'}}$ is independent of $n$ and $n'$.
It turns out that such frames have minimal coherence~\eqref{equation.definition of coherence}.
In short, letting $\set{\bfphi_n}_{n\in\calN}$ be any unit vectors in $\bbC^{\calM}$, we have
\begin{align}
\notag
0
&\leq\Tr[(\bfPhi\bfPhi^*-\tfrac NM\bfI)^2]\\
\notag
&=\Tr[(\bfPhi^*\bfPhi)^2]-\tfrac{N^2}M\\
\notag
&=\sum_{n\in\calN}\sum_{n'\in\calN}\abs{\ip{\bfphi_n}{\bfphi_{n'}}}^2-\tfrac{N^2}M\\
\label{equation.derivation of Welch bound}
&\leq N+N(N-1)\mu^2-\tfrac{N^2}M.
\end{align}
Solving for $\mu$ yields the Welch bound~\eqref{equation.Welch bound}.
Moreover, equality in~\eqref{equation.Welch bound} forces equality throughout~\eqref{equation.derivation of Welch bound}.
In particular, any unit norm vectors $\set{\bfphi_n}_{n\in\calN}$ which achieve the Welch bound are necessarily equiangular (since $\abs{\ip{\bfphi_n}{\bfphi_{n'}}}=\mu$ for all $n\neq n'$), and also a tight frame, since the Frobenius norm of $\bfPhi\bfPhi^*-\frac NM\bfI$ is zero, forcing $\bfPhi\bfPhi^*=\frac NM\bfI$.
Conversely, any ETF is both tight and equiangular, yielding equality in the first and last inequalities in~\eqref{equation.derivation of Welch bound}, respectively, and thus equality in~\eqref{equation.Welch bound}.

In the next section, we convert some of the Steiner ETFs of~\cite{FickusMT12} into constant-amplitude ETFs.
In general, Steiner ETFs are built from special types of balanced incomplete block designs known as \textit{$(2,K,V)$-Steiner systems}.
Such a system consists of a set of $V$ points $\calV$ along with a set of $B$ blocks (subsets) of $\calV$, denoted $\calB$, with the property that every block contains exactly $K$ points, every point is contained in exactly $R$ blocks, and every pair of points is contained in exactly one block.
Ordering the points and blocks, we can form a real, $\set{0,1}$-valued $B\times V$ \textit{incidence matrix} $\bfB$ that indicates which points belong to which blocks; the rows of this matrix must sum to $K$, its columns must sum to some constant number $R$, and the dot product of any two distinct columns must be $1$.
For example, consider the $(2,2,4)$-Steiner system that consists of all $2$-element subsets of a set of $4$ elements:
\begin{equation}
\label{equation.6x16 B}
\bfB
=\begin{bmatrix}
1&1&0&0\\
0&0&1&1\\
1&0&1&0\\
0&1&0&1\\
1&0&0&1\\
0&1&1&0
\end{bmatrix}.
\end{equation}
Here, $B=6$ and $R=3$.

As shown in~\cite{FickusMT12}, every $(2,K,V)$-Steiner system generates an ETF of $N=V(R+1)$ vectors in a space of dimension $M=B$.
The main idea is to take a tensor-like combination of $\bfB$ with a unimodular regular simplex of $R+1$ vectors in $R$-dimensional space.
For example, for the $(2,2,4)$-Steiner system~\eqref{equation.6x16 B} in which $R=3$, we can construct such a simplex by removing a row from a $4\times 4$ matrix with orthogonal columns and unimodular entries, such as a DFT or Hadamard matrix:
\begin{equation}
\label{equation.6x16 F}
\bfF
=\begin{bmatrix}
+&-&+&-\\
+&+&-&-\\
+&-&-&+
\end{bmatrix}.
\end{equation}
Here and throughout, ``$+$" and ``$-$" denote $1$ and $-1$, respectively.
To construct an ETF from $\bfB$ and $\bfF$, we replace each of the $R$ nonzero entries in any given column of $\bfB$ with a corresponding row from $\bfF$, and replace each zero entry of $\bfB$ with a $(R+1)$-long row of zeros.
We then normalize the resulting columns.
In particular, Figure~\ref{figure.6x16 Phi} gives the $6\times 16$ ETF $\bfPhi$ obtained by ``tensoring"~\eqref{equation.6x16 B} with~\eqref{equation.6x16 F} in this fashion.
\begin{figure*}
\vspace{-\bigskipamount}
\begin{equation*}
\bfPhi=\frac1{\sqrt{3}}\left[\begin{array}{cccc|cccc|cccc|cccc}
+&-&+&-&+&-&+&-&0&0&0&0&0&0&0&0\\
0&0&0&0&0&0&0&0&+&-&+&-&+&-&+&-\\
+&+&-&-&0&0&0&0&+&+&-&-&0&0&0&0\\
0&0&0&0&+&+&-&-&0&0&0&0&+&+&-&-\\
+&-&-&+&0&0&0&0&0&0&0&0&+&-&-&+\\
0&0&0&0&+&-&-&+&+&-&-&+&0&0&0&0
\end{array}\right].
\end{equation*}
\caption{\label{figure.6x16 Phi}A Steiner ETF of $16$ vectors in $6$-dimensional space obtained by ``tensoring" the incidence matrix \eqref{equation.6x16 B} of a $(2,2,4)$-Steiner system with a regular simplex of $4$ vectors in $3$-dimensional space~\eqref{equation.6x16 F} according to~\cite{FickusMT12}.
Here, the Welch bound~\eqref{equation.Welch bound} is $1/3$, and any two distinct columns have a dot product of this magnitude.
Indeed, grouping the columns as $V=4$ sets of $R+1=4$ vectors (as pictured), any two distinct columns from the same set (like the first and second columns of $\bfPhi$) have a dot product of $-1/3$ since it corresponds to a dot product of distinct columns in our regular simplex~\eqref{equation.6x16 F}.
Meanwhile, any two columns from distinct sets (like the first and fifth columns of $\bfPhi$) have only one point of common support, since any two distinct points in our Steiner system determine a unique block;
as such, the dot product of such columns has value $\pm1/3$.
Such ETFs are sparse---many of the entries of $\bfPhi$ are zero---and also have low spark: the first four of these vectors are linearly dependent.
In this paper, we show how to unitarily transform this matrix into a constant-amplitude ETF, a trick which ``corrects" its sparsity but not its spark.
Our approach will only work due to the fact that our Steiner system here is resolvable: the first and second blocks (rows) in~\eqref{equation.6x16 B} yield a partition of the points (columns), as do its third and fourth blocks, and its fifth and sixth blocks.
}
\end{figure*}

Since any finite-dimensional space always contains a unimodular regular simplex, the only restrictions on the existence of such ETFs arise from restrictions on Steiner systems themselves.
For example, the $B$ and $R$ parameters of a $(2,K,V)$-Steiner system are uniquely determined by $K$ and $V$ according to the necessary relationships that
\begin{equation}
\label{equation.Steiner parameter relationships}
BK=VR,\quad R(K-1)=V-1.
\end{equation}
The first identity follows from counting the total number of $1$'s in the incidence matrix $\bfB$ both row-wise and column-wise; the second follows from counting the number of $1$'s in $\bfB$ that lie to the right of a $1$ in the first column.
In particular, for a $(2,K,V)$-Steiner system to exist, both $R=(V-1)/(K-1)$ and $B=V(V-1)/[K(K-1)]$ must be integers.

Our parameters must also satisfy $B\geq V$; known as \textit{Fisher's inequality}, this follows from the fact that the $V\times V$ matrix $\bfB^\rmT \bfB$ is necessarily of full rank, since its off-diagonal entries are $1$ while its diagonal entries are $R=(V-1)/(K-1)>1$.
These facts are important here since the parameters $K$ and $V$ indicate the dimensions of the resulting Steiner ETF.
Indeed, in light of~\eqref{equation.Steiner parameter relationships}, the redundancy of a Steiner ETF is
\begin{equation*}
\tfrac{N}{M}
=\tfrac{V(R+1)}{B}
=K\tfrac{V(R+1)}{BK}
=K\tfrac{V(R+1)}{VR}
=K(1+\tfrac1R).
\end{equation*}
Since~\eqref{equation.Steiner parameter relationships} and Fisher's inequality give $K/R=V/B\leq 1$, the redundancy of a Steiner ETF is essentially $K$.
Moreover, for any fixed $K$, both $M$ and $N$ grow quadratically with $V$:
\begin{equation*}
M=B=\tfrac{V(V-1)}{K(K-1)},
\quad
N=V(R+1)=V(\tfrac{V-1}{K-1}+1).
\end{equation*}
In particular, in order to build ETFs of various sizes and redundancies, we need explicit constructions of $(2,K,V)$-Steiner systems which permit flexible, independent control of both $K$ and $V$.
There are three known families of such systems~\cite{ColbournM07}, all arising from finite geometry.

To be precise, for any prime power $q$ and $j\geq 1$, there exist affine geometry-based Steiner systems with $K=q$ and $V=q^{j+1}$.
For $j\geq 2$, there also exist projective geometry-based systems with $K=q+1$ and $V=(q^{j+1}-1)(q-1)$.
In either case, varying $q$ and $j$ controls the redundancy and size of the ETF, respectively.
The third family is \textit{Denniston designs}~\cite{ColbournM07} in which $K=2^i$, $V=2^{i+j}+2^i-2^j$ for some $2\leq i<j$ which arise from maximal arcs in projective spaces~\cite{Denniston69}.
Importantly, both the affine and Denniston designs are resolvable~\cite{FurinoMY96}.
As we now discuss, this means we can transform them into constant-amplitude ETFs.

\section{Kirkman ETFs}
\label{section.Kirkman ETFs}

In this section, we introduce a method for unitarily transforming certain Steiner ETFs, like the one depicted in Figure~\ref{figure.6x16 Phi}, into constant-amplitude ETFs.
This method requires the underlying Steiner system to be \textit{resolvable}, meaning its blocks $\calB$ can be partitioned into disjoint subcollections $\set{\calB_r}_{r\in\calR}$ so that the blocks in any given $\calB_r$ form a partition for $\calV$.
For example, the $(2,2,4)$-Steiner system given in~\eqref{equation.6x16 B} is resolvable: its first and second blocks (rows) form a partition for our underlying set of $V=4$ points, as do its third and fourth, and its fifth and sixth.
The main idea of this new method is to multiply the synthesis matrix of a resolvable Steiner ETF, like Figure~\ref{figure.6x16 Phi}, by a block-Hadamard/DFT matrix to obtain a constant-amplitude ETF; see Figure~\ref{figure.6x16 Psi}.
\begin{figure*}
\begin{align*}
\bfPsi&=\frac1{\sqrt{2}}\left[\begin{array}{cc|cc|cc}
+&+&0&0&0&0\\
+&-&0&0&0&0\\
\hline
0&0&+&+&0&0\\
0&0&+&-&0&0\\
\hline
0&0&0&0&+&+\\
0&0&0&0&+&-
\end{array}\right]\times
\frac1{\sqrt{3}}\left[\begin{array}{cccc|cccc|cccc|cccc}
+&-&+&-&+&-&+&-&0&0&0&0&0&0&0&0\\
0&0&0&0&0&0&0&0&+&-&+&-&+&-&+&-\\
\hline
+&+&-&-&0&0&0&0&+&+&-&-&0&0&0&0\\
0&0&0&0&+&+&-&-&0&0&0&0&+&+&-&-\\
\hline
+&-&-&+&0&0&0&0&0&0&0&0&+&-&-&+\\
0&0&0&0&+&-&-&+&+&-&-&+&0&0&0&0
\end{array}\right]\\
&=\frac1{\sqrt{6}}\left[\begin{array}{cccccccccccccccccc}
+&-&+&-&+&-&+&-&+&-&+&-&+&-&+&-\\
+&-&+&-&+&-&+&-&-&+&-&+&-&+&-&+\\
+&+&-&-&+&+&-&-&+&+&-&-&+&+&-&-\\
+&+&-&-&-&-&+&+&+&+&-&-&-&-&+&+\\
+&-&-&+&+&-&-&+&+&-&-&+&+&-&-&+\\
+&-&-&+&-&+&+&-&-&+&+&-&+&-&-&+
\end{array}\right]
\end{align*}
\caption{\label{figure.6x16 Psi}
Constructing a constant-amplitude ETF of 16 vectors in $6$-dimensional space by multiplying the Kirkman (resolvable Steiner) ETF of Figure~\ref{figure.6x16 Phi} by a unitary block-DFT/Hadamard matrix.
Here, the horizontal lines in $\bfPhi$ indicate the way in which the $B=6$ blocks (rows) of the underlying Steiner system~\eqref{equation.6x16 B} can be broken up into $R=3$ distinct partitions of a set of $V=4$ elements.
We obtain a new ETF $\bfPsi$ by multiplying $\bfPhi$ on the left by a unitary block-diagonal matrix consisting of $R=3$ DFT/Hadamard matrices.
Multiplying these two matrices blockwise, we see that the resulting ETF has constant-amplitude: every entry of $\bfPsi$ is a product of exactly one nonzero entry of $\bfPhi$ with an entry of a DFT/Hadamard matrix.
Having constant-amplitude, this ETF $\bfPsi$ is better suited than $\bfPhi$ for certain radio waveform design problems like CDMA~\cite{HeathTDS04}.
Nonetheless, the second ETF is obtained by applying a unitary operator to the first, meaning they share many of the same linear algebraic properties.
In particular, since the first four columns of $\bfPhi$ are linearly dependent, the first four columns of $\bfPsi$ are as well.
In Section~\ref{section.Kirkman vs harmonic}, we show that an important subclass of harmonic ETFs, namely those that arise from McFarland  difference sets, can be built with this approach, and thus are unitarily equivalent to very sparse Steiner ETFs.
This allows us, for the first time, to observe that such ETFs do not satisfy compressed sensing's RIP to any degree that rivals random constructions.
}
\end{figure*}

Not every Steiner system is resolvable.
Indeed, if \textit{any} subset of the blocks $\calB$ forms a partition of $\calV$, then $K$ must divide $V$: each block contains $K$ points and there are $V$ points total.
This requirement alone prohibits the famous $(2,3,7)$-Steiner system known as the \textit{Fano plane} from being resolvable.
When coupled with the previous restriction that $K-1$ divides $V-1$, this new condition subsumes the previous requirement that $K(K-1)$ divides $V(V-1)$.
Moreover, since we necessarily have $V\equiv 1\bmod K-1$ and $V\equiv0\bmod K$ where $K$ is relatively prime to $K-1$, the Chinese Remainder Theorem gives that these two conditions are equivalent to having $V\equiv K\bmod K(K-1)$.
For resolvable designs, it also turns out~\cite{FurinoMY96} that Fisher's inequality can be strengthened to \textit{Bose's condition} that $B\geq V+R-1$.

Nevertheless, many Steiner systems are resolvable, such as those arising from affine geometries over finite fields and Denniston designs~\cite{FurinoMY96}.
It seems to be an open question whether or not projective geometries with $K=q+1$ and $V=(q^{j+1}-1)/(q-1)$ are resolvable when $j$ is odd.
At least is some cases, the answer is yes: when $q=2$ and $j=2$, this is~\textit{Kirkman's schoolgirl problem}.
Since this famous problem is so closely associated with resolvable Steiner systems, we refer to the constant-amplitude ETFs that arise from such systems as \textit{Kirkman ETFs}.

We now formally verify that every resovable Steiner system generates a (constant-amplitude) Kirkman ETF that is unitarily equivalent to a (sparse) Steiner ETF.
Here, as usual, the quickest way to verify that certain vectors form an ETF is to show they satisfy the Welch bound~\eqref{equation.Welch bound} with equality.
In this Steiner-system-induced setting where $M=B$ and $N=V(R+1)$, the lower bound itself is simply $1/R$; noted in~\cite{FickusMT12},
this can be most easily seen by making repeated use of the identities~\eqref{equation.Steiner parameter relationships} to show $M(N-1)/(N-M)=R^2$.
This is a special case of a known necessary integrality condition~\cite{SustikTDH07}; if all the entries in an ETF are suitably-normalized roots of unity, then $M(N-1)/(N-M)$ is necessarily an integer.

Before stating the result, it is helpful to introduce some notation.
Note that in any resolvable $(2,K,V)$-Steiner system, the number of blocks in any single partition is $V/K$.
Since the total number of blocks is $B=VR/K$, the number of distinct partitions of $\calV$ is $R$.
As such, we enumerate these partitions using some $R$-element indexing set $\calR$, and write our blocks as the disjoint union $\calB=\sqcup_{r\in\calR}\calB_r$.
Here, for any $r\in\calR$, the $S:=V/K$ blocks that lie in $\calB_r$ form a partition for $\calV$, and we index them with some $S$-element indexing set $\calS$.
To be precise, for any $r\in\calR$, let $\calB_r=\set{b_{r,s}}_{s\in\calS}$ where $\calV=\sqcup_{s\in\calS}b_{r,s}$.
We now state and prove our first main result:
\begin{theorem}
\label{theorem.Kirkman construction}
Let $(\calV,\calB)$ be a resolvable $(2,K,V)$-Steiner system:
let $\set{\calB_r}_{r\in\calR}$ be a partition of $\calB$ where for any $r$, $\calB_r=\set{b_{r,s}}_{s\in\calS}$ is a partition of $\calV$.
Let $\set{\bff_{u}}_{u=0}^{R}$ be a unimodular regular simplex in $\bbC^\calR$ and let $\set{\bfh_{s}}_{s\in\calS}$ be a unimodular orthogonal basis for $\bbC^{\calS}$.
Then letting $\calM=\calR\times\calS$ and $\calN=\calV\times\set{0,\dotsc,R}$, the $V(R+1)$ vectors $\set{\bfphi_{u,v}}_{(u,v)\in\calN}$ form a Steiner ETF for the $B$-dimensional space $\bbC^{\calM}$:
\begin{equation}
\label{equation.definition of resolvable Steiner ETF}
\bfphi_{u,v}(r,s)
:=R^{-\frac12}\left\{\begin{array}{cl}\bff_{u}(r),&v\in b_{r,s},\\0,&v\notin b_{r,s}.\end{array}\right.
\end{equation}
Moreover, applying a unitary operator to this Steiner ETF yields the Kirkman ETF $\set{\bfpsi_{u,v}}_{(u,v)\in\calN}$ defined by
\begin{equation}
\label{equation.definition of Kirkman ETF}
\bfpsi_{u,v}(r,s)
:=B^{-\frac12}\bff_{u}(r)\bfh_{s(r,v)}(s),
\end{equation}
where for any $r\in\calR$ and $v\in\calV$, $s(r,v)$ denotes the unique $s\in\calS$ such that $v\in b_{r,s}$.
\end{theorem}

\begin{proof}
We first prove that $\set{\bfphi_{u,v}}_{(u,v)\in\calN}$ is an ETF.
To do this, it suffices to show that each $\bfphi_{u,v}$ is unit norm and that the inner product of any distinct two of these vectors has modulus equal to the Welch bound $1/R$.
In general, we have
\begin{equation}
\label{equation.proof that Kirkman yields ETFs 1}
\ip{\bfphi_{u,v}}{\bfphi_{u',v'}}
=\sum_{r\in\calR}\sum_{s\in\calS}\bfphi_{u,v}(r,s)[\bfphi_{u',v'}(r,s)]^*.
\end{equation}
In the special case where $v=v'$, note that for any $r\in\calR$, the fact that $\set{b_{r,s}}_{s\in\calS}$ is a partition of $\calV$ implies there is exactly one $s\in\calS$ such that $v\in b_{r,s}$.
In light of~\eqref{equation.definition of resolvable Steiner ETF}, this fact reduces \eqref{equation.proof that Kirkman yields ETFs 1} in this case to
\begin{equation*}
\ip{\bfphi_{u,v}}{\bfphi_{u',v}}
=\tfrac1R\sum_{r\in\calR}\bff_{u}(r)[\bff_{u'}(r)]^*
=\tfrac1R\ip{\bff_{u}}{\bff_{u'}}.
\end{equation*}
When coupled with the fact that $\set{\bff_{u}}_{u=0}^{R}$ is a unimodular regular simplex, this implies that the $\bfphi_{u,v}$'s have unit norm and satisfy $\abs{\ip{\bfphi_{u,v}}{\bfphi_{u',v}}}=1/R$ whenever $u\neq u'$, as needed.
To show that we also have $\abs{\ip{\bfphi_{u,v}}{\bfphi_{u',v'}}}=1/R$ whenever $v\neq v'$, recall that since $(\calV,\calB)$ is a $(2,K,V)$-Steiner system, there is exactly one block $b=b_{r_0,s_0}$ that contains both $v$ and $v'$.
Again recalling~\eqref{equation.definition of resolvable Steiner ETF}, this means that there is only one nonzero summand of \eqref{equation.proof that Kirkman yields ETFs 1}, yielding
\begin{equation*}
\abs{\ip{\bfphi_{u,v}}{\bfphi_{u',v'}}}
=\tfrac1R\abs{\bff_{u}(r_0)}\abs{\bff_{u'}(r_0)}
=\tfrac1R.
\end{equation*}
Thus, $\set{\bfphi_{u,v}}_{(u,v)\in\calN}$ is an ETF, as claimed.
For the second conclusion, note that by~\eqref{equation.definition of Kirkman ETF}, $\ip{\bfpsi_{u,v}}{\bfpsi_{u',v'}}$ is:
\begin{align}
\notag
&\tfrac1B\sum_{r\in\calR}\sum_{s\in\calS}\bff_{u}(r)\bfh_{s(r,v)}(s)[\bff_{u'}(r)\bfh_{s(r,v')}(s)]^*\\
\label{equation.proof that Kirkman yields ETFs 2}
&\quad=\tfrac1B\sum_{r\in\calR}\bff_{u}(r)[\bff_{u'}(r)]^*\ip{\bfh_{s(r,v)}}{\bfh_{s(r,v')}}.
\end{align}
Fixing $r$ for the moment, note that since $\set{\bfh_{s}}_{s\in\calS}$ is a unimodular orthogonal basis for the space $\bbC^{\calS}$ of dimension $S=V/K=B/R$, we have that $\ip{\bfh_{s(r,v)}}{\bfh_{s(r,v')}}=B/R$ when $s(r,v)=s(r,v')$ and is otherwise zero.
Since $s(r,v)$ denotes the unique $s\in\calS$ such that $v\in b_{r,s}$, this implies
\begin{equation}
\label{equation.proof that Kirkman yields ETFs 3}
\ip{\bfh_{s(r,v)}}{\bfh_{s(r,v')}}
=\tfrac BR\sum_{s\in\calS}\bfone_{b_{r,s}}(v)\bfone_{b_{r,s}}(v'),
\end{equation}
where $\bfone_{b_{r,s}}:\calV\rightarrow\set{0,1}\subseteq\bbC$ is the characteristic function of the block $b_{r,s}$.
For every $r\in\calR$, substituting~\eqref{equation.proof that Kirkman yields ETFs 3} into~\eqref{equation.proof that Kirkman yields ETFs 2} and then recalling~\eqref{equation.definition of resolvable Steiner ETF} gives that
\begin{align}
\notag
\ip{\bfpsi_{u,v}}{\bfpsi_{u',v'}}
&=\tfrac1R\sum_{r\in\calR}\sum_{s\in\calS}\bff_{u}(r)[\bff_{u'}(r)]^*\bfone_{b_{r,s}}(v)\bfone_{b_{r,s}}(v')\\
\notag
&=\sum_{r\in\calR}\sum_{s\in\calS}\bfphi_{u,v}(r,s)[\bfphi_{u',v'}(r,s)]^*\\
\label{equation.proof that Kirkman yields ETFs 4}
&=\ip{\bfphi_{u,v}}{\bfphi_{u',v'}}.
\end{align}
Since~\eqref{equation.proof that Kirkman yields ETFs 4} holds for all $(u,v),(u',v')\in\calN$, we know that $\set{\bfpsi_{u,v}}_{(u,v)\in\calN}$ is also an ETF, and moreover, is obtained by applying a unitary transformation to $\set{\bfphi_{u,v}}_{(u,v)\in\calN}$.
In truth, this transformation is a block-unitary transform, but we do not need this specificity for the work that follows.
\end{proof}

For the remainder of this section, we consider the ramifications of Theorem~\ref{theorem.Kirkman construction} on the existence of constant-amplitude ETFs.
In particular, we first describe Kirkman ETFs that arise from known \textit{flexible} families of resolvable $(2,K,V)$-Steiner systems, meaning they permit independent control of $K$ and $V$.
We then describe some~\textit{inflexible} families, meaning that $K$ uniquely determines $V$, or vice versa.
Finally, we conclude this section with a discussion of the known asymptotic existence results for resolvable Steiner systems.

For each of these families, we state whether or not constant-amplitude ETFs with those parameters have been found before.
To be clear, the existence of Steiner ETFs of these sizes was already noted in~\cite{FickusMT12}.
However, before Theorem~\ref{theorem.Kirkman construction}, the only known method for constructing constant-amplitude ETFs was to use difference sets~\cite{XiaZG05,DingF07}.
We also do our best to answer a deeper question: whether or not a given Kirkman ETF is actually a harmonic ETF in disguise.
Note that this would necessarily imply that there exist difference sets $\calD$ of $M=B$ elements in Abelian groups $\calG$ of order $N=V(R+1)$.
This, in turn, requires that
\begin{equation*}
\Lambda
:=\tfrac{M(M-1)}{N-1}
=\tfrac{V(V-K)}{K^2(K-1)}
\end{equation*}
is an integer, since $\Lambda$ is the number of times any nonzero element of $\calG$ may be written as a difference of two elements in $\calD$.
However, every Kirkman ETF automatically satisfies this integrality condition: if a resolvable $(2,K,V)$-Steiner system exists, then $V\equiv K\bmod K(K-1)$; writing $V=WK(K-1)+K$ gives $\Lambda=W[W(K-1)+1]$.
Moreover, this implies that the \textit{degree} $M-\Lambda$ of such a difference set is necessarily the perfect square $M-\Lambda=[W(K-1)+1]^2$, meaning that the necessary conditions of the Bruck-Ryser-Chowla Theorem are automatically satisfied whenever $N$ is even~\cite{FurinoMY96}.
As such, trying to show that a given Kirkman ETF is \textit{not} harmonic can quickly lead to hard, open problems concerning the existence of difference sets.

\subsection{Flexible Kirkman ETFs}

\subsubsection{Affine geometries over finite fields}
For any $j\geq 1$ and prime power $q$, there exists a resolvable $(2,K,V)$-Steiner system with $K=q$ and $V=q^{j+1}$.
Here, the points $\calV$ in this design are the vectors in $\bbF_q^{j+1}$ where $\bbF_q$ is the finite field of order $q$.
Meanwhile, the blocks $\calB$ are affine lines in this space, namely sets of the form $\set{au+v : a\in\bbF_q}$ for some direction vector $u\in\bbF_q^{j+1}\backslash\set{0}$ and initial point $v\in\bbF_q^{j+1}$.
These systems play an important role in the theory of the next section, and we describe them more fully there.
For now, the most important things to note are that (i) these systems are easy to construct explicitly, meaning the construction of Theorem~\ref{theorem.Kirkman construction} can be truly implemented; (ii) the resulting Kirkman ETFs consist of $N$ vectors in an $M$-dimensional space where
\begin{equation}
\label{equation.affine Kirkman parameters}
M=q^j\biggparen{\frac{q^{j+1}-1}{q-1}},
\quad
N=q^{j+1}\biggparen{\frac{q^{j+1}-1}{q-1}+1};
\end{equation}
and that (iii) the redundancy and size of this ETF can be controlled by manipulating $q$ and $j$, respectively.
The existence of constant-amplitude ETFs with these parameters is not new: harmonic ETFs with dimensions~\eqref{equation.affine Kirkman parameters} can be constructed with McFarland difference sets~\cite{DingF07}.
In the next section, we show this is not a coincidence: we prove that every McFarland harmonic ETF is a Kirkman ETF, and as such, is unitarily equivalent to a low-spark, sparse Steiner ETF.

\subsubsection{Denniston designs}
For any positive integers $i$ and $j$, $i\leq j$, there exists a resolvable $(2,K,V)$-Steiner system with $K=2^i$ and $V=2^{i+j}+2^i-2^j$.
The construction is nontrivial: one constructs a \textit{maximal arc} in the projective plane of order $2^j$ using an irreducible quadratic form~\cite{Denniston69}, and then constructs a resolvable design in terms of this arc~\cite{FurinoMY96}.
The resulting Kirkman ETF has $M=(2^j+1)(2^j+1-2^{j-i})$ vectors in a space of dimension $N=2^i(2^j+2)(2^j+1-2^{j-i})$.

Note that when $i=j$, these designs have the same parameters as an affine geometry where $q=2^i$.
Meanwhile, for $i<j$, the constant-amplitude ETFs generated by these designs seem to be new.
For example, when $i=2$ and $j=3$, we find that there exists a constant-amplitude ETF of $N=280$ vectors in a space of dimension $M=63$; such an ETF is not found in the existing literature~\cite{XiaZG05,DingF07}.
Such ETFs might be harmonic: we did not find any examples of $i$ and $j$ for which it is known that there cannot exist an $M$-element difference set in an Abelian group of order $N$; since $N$ is even, the Bruck-Ryser-Chowla Theorem is toothless.
We leave a more thorough investigation of this problem for future research.

\subsection{Inflexible Kirkman ETFs}

\subsubsection{Round-robin tournaments}

For any positive integer $V$, consider the $(2,2,V)$-Steiner system that consists of every two-element subset of $\calV=\set{1,\dotsc,V}$, such as the $(2,2,4)$-Steiner system whose incidence matrix is given in~\eqref{equation.6x16 B}.
When $V$ is even, this system is resolvable via the famous \textit{round-robin} schedule, which is sometimes used in tournament competitions, as it ensures that each competitor faces all others exactly once while letting the entire tournament be as quick as possible.
The resulting family of constant-amplitude Kirkman ETFs is inflexible, since the redundancy $N/M$ of any such frame is essentially two: $M=V(V-1)/2$, $N=V^2$.

Some of the constant-amplitude ETFs generated by these designs via Theorem~\ref{theorem.Kirkman construction} are new.
To be clear, when $V=2^{j+1}$, ETFs with these parameters are well-known~\cite{DingF07}, arising from McFarland difference sets in Abelian groups isomorphic to $\bbZ_2^{2j+2}$.
However, when $V$ is even but not a power of $2$, some of the resulting Kirkman ETFs do not arise from difference sets.
In particular, there does not exist a difference set of $M=45$ elements in an Abelian group of order $N=100$~\cite{JungnickelPS07}, and so the $(2,2,10)$-Round Robin Kirkman ETF is not harmonic.
In Section~\ref{section.Grey-Rankin}, we exploit these ideas to build new examples of \textit{real-valued} constant-amplitude ETFs provided there exists a Hadamard matrix of size $V$; this leads to new examples of (nonlinear) binary codes that achieve the Grey-Rankin bound.

\subsubsection{Kirkman's Schoolgirl Problem}

For any positive integer $V\equiv 3\bmod 6$, there exists a resolvable $(2,3,V)$-Steiner (triple) system~\cite{RayW71}.
The resulting Kirkman ETFs have an approximate redundancy of $3$, consisting of $V(V+1)/2$ vectors in a space of dimension $V(V-1)/6$.
At least some of these frames are new constant-amplitude ETFs: when $V=15$, for example (Kirkman's original problem), the resulting ETF consists of $120$ vectors in a space of dimension $35$, and there does not exist an Abelian difference set with those paramters~\cite{JungnickelPS07}.

\subsubsection{Three-dimensional projective geometries}

For any prime power $q$, a resolvable $(2,q+1,q^3+q^2+q+1)$-Steiner system exists~\cite{Lorimer74}.
The resulting family of Kirkman ETFs is inflexible: though both $M=(q^2+1)(q^2+q+1)$ and $N=(q^2+q+2)(q^3+q^2+q+1)$ can grow arbitrarily large, the single parameter $q$ determines both the size and redundancy of such frames.
Note that when $q=2$ this design is a $(2,3,15)$-Steiner triple system which, as noted above, generates an ETF which is not harmonic.
As such, at least some of these frames are new constant-amplitude ETFs.

To be clear, projective geometries generate a flexible family of \textit{Steiner} ETFs: for any $j\geq2$
the projective geometry with parameters $K=q+1$ and $V=(q^{j+1}-1)/(q-1)$ generates a Steiner ETF~\cite{FickusMT12} with dimensions
\begin{equation*}
M=\frac{(q^j-1)(q^{j+1}-1)}{(q+1)(q-1)^2},
\quad
N=\frac{q^{j+1}-1}{q-1}\biggparen{\frac{q^{j}-1}{q-1}-1},
\end{equation*}
and varying $q$ and $j$ independently generates ETFs of various sizes and redundancies.
However, we could only find references to projective geometries being resolvable---and thus able to generate Kirkman ETFs---in special cases, like here where $j=3$.
Note that in order to be resolvable we need $K$ to divide $V$ which requires $j$ to be odd; it seems to be an open question whether such systems are resolvable for odd $j\geq 5$.

\subsubsection{Unitals}

For any prime power $q$, there exists a resolvable $(2,q+1,q^3+1)$-Steiner system~\cite{Bose58}.
This family, like the last, is inflexible, since $q$ determines both $M=q^2(q^2-q+1)$ and $N=(q^2+1)(q^3+1)$.
At least some of these are new constant-amplitude ETFs: taking $q=3$ yields a constant-amplitude ETF of $280$ vectors in a space of dimension $63$ which, as noted earlier, is not in the literature.
We did not find any examples of any such ETFs which are provably not harmonic;
note that whenever $q$ is odd, $N$ is even and so the conditions of the Bruck-Ryser-Chowla Theorem are automatically satisfied.

\subsection{Existence results for Kirkman ETFs}

A remarkable fact about resolvable Steiner systems is that they are known to exist asymptotically.
That is, for any $K\geq 2$ it is know that there exists a positive integer $V_0(K)$ such that for any $V\geq V_0(K)$ which satisfies $V\equiv K\bmod K(K-1)$, there exists a resolvable $(2,K,V)$-Steiner system~\cite{RayW73}.
Thus, for any such $V$, we know there exists a constant-amplitude ETF of $N=V(V-K+2)/(K-1)$ vectors in a space of dimension $M=V(V-1)/[K(K-1)]$.
In particular, for any $K\geq 2$, there exist constant-amplitude ETFs whose redundancy $N/M$ is arbitrarily close to $K$.
In contrast, all known examples of harmonic ETFs have redundancies which are essentially a power of a prime.
Unfortunately, the methods used to demonstrate the existence of these systems are not constructive in any practical sense.
Nevertheless, these existence results encourage the search for explicit constructions of constant-amplitude ETFs with arbitrary redundancies.

\section{Connecting Kirkman ETFs and harmonic ETFs}
\label{section.Kirkman vs harmonic}

In this section, we show that an important class of harmonic ETFs, namely those generated by McFarland difference sets, can also be constructed as Kirkman ETFs via Theorem~\ref{theorem.Kirkman construction}.
As a corollary, we find that harmonic ETFs from this particular family have Steiner representations, making them less desirable for certain compressed-sensing-related applications.
This result also demonstrates how truly rare it is to discover a new flexible family of ETFs.

To be precise, as noted in the introduction, there are only three known approaches for constructing nontrivial ETFs: via conference matrices~\cite{StrohmerH03}, difference sets~\cite{Strohmer08,XiaZG05,DingF07} and Steiner systems~\cite{FickusMT12}.
In the previous section, we refined the Steiner-based approach so as to produce constant-amplitude Kirkman ETFs.
Moreover, nearly all of the particular instances of these constructions are inflexible.
Indeed, ETFs generated by conference matrices have redundancy $\frac NM=2$;
surveying~\cite{XiaZG05,DingF07} as well as a comprehensive list of Abelian difference sets~\cite{JungnickelPS07}, we see that harmonic ETFs generated by Paley, Hadamard, twin prime power and Davis-Jedwab-Chen difference sets all have an approximate redundancy of $2$, while those generated by other cyclotomic or Spence difference sets have approximate redundancies of either $3$, $4$ or $8$;
as noted earlier, families of Steiner systems with a fixed $K$ yield ETFs whose approximate redundancy is $K$, and those produced from unital designs are also inflexible.

There are only five known flexible families of ETFs: harmonic ETFs arising from (i) Singer difference sets~\cite{XiaZG05} and (ii) McFarland difference sets~\cite{DingF07}, and Steiner ETFs arising from (iii) affine geometries, (iv) Denniston designs and (v) projective geometries~\cite{FickusMT12}.
Also, as shown in the previous section, classes (iii) and (iv) arise from resolvable Steiner systems, meaning we can apply a unitary operator to them to produce constant-amplitude ETFs.
In this section, we show that modulo such unitaries, (ii) is a special case of (iii).
That is, we show that in truth there are only four known flexible families of ETFs;
three of these four, namely (i), (ii/iii) and (iv), have constant-amplitude representations;
another three of these four, namely (ii/iii), (iv) and (v), have very sparse representations.

To show that every McFarland harmonic ETF is a special example of an affine Kirkman ETF, let $j\geq1$, let $q$ be a prime power, and consider $\calG\times\calV$ where $\calG$ is any Abelian group of order $(q^{j+1}-1)/(q-1)+1$ and $\calV$ is the additive group of finite field $\bbF_{q^{j+1}}$.
We form a harmonic ETF over this group via the approach of~\cite{DingF07} by letting $\calD$ be a McFarland difference set in $\calG\times\calV$.
McFarland's approach~\cite{McFarland73} is clever, and is eerily similar to Goethals and Seidel's method for constructing strongly regular graphs from Steiner systems~\cite{GoethalsS70}; though made independently, we show these two approach are related, since both can lead to the same ETFs.

The first step in forming a McFarland difference set is to parametrize the distinct hyperplanes in $\bbF_{q^{j+1}}$, regarded as the $(j+1)$-dimensional vector space $\bbF_q^{j+1}$ over the field $\bbF_q$.
Note that each hyperplane is the null space of some nontrivial linear functional over $\bbF_{q}^{j+1}$.
Moreover, there are $q^{j+1}-1$ such functionals, since each can be uniquely represented by a nonzero $1\times(j+1)$ matrix.
Also, any two such null spaces are equal precisely when their corresponding functionals are nonzero scalar multiples of each other.
As such, there are $R:=(q^{j+1}-1)/(q-1)$ distinct hyperplanes overall.

To find an explicit expression for these hyperplanes, regard $\bbF_{q^{j+1}}$ as an extension of $\bbF_q$, and let $\tr_{q^j/q}:\bbF_{q^{j+1}}\rightarrow\bbF_q$ be the associated \textit{field trace}, namely the sum of the automorphisms of $\bbF_{q^{j+1}}$ that fix $\bbF_q$.
Regarding $\bbF_{q^{j+1}}$ as the vector space $\bbF_q^{j+1}$, it is well known that this trace is a nontrivial linear functional, and so its null space
\begin{equation}
\label{equation.McFarland vs affine 0}
\calS
=\set{v\in\bbF_{q^{j+1}}: \tr_{q^j/q}(v)=0}
\end{equation}
is one example of a hyperplane in $\bbF_q^{j+1}$, and thus has cardinality $S=q^j$.
To find the remaining hyperplanes, let $\gamma$ be a primitive element of $\bbF_{q^{j+1}}$, meaning $\gamma$ generates its cyclic multiplicative group.
Since the mappings $v\mapsto\tr_{q^j/q}(\gamma^{-r}v)$ are distinct for every $r=0,\dotsc,q^{j+1}-1$, every nontrivial linear functional on $\bbF_q^{j+1}$ can be represented this way.
Moreover, since the nonzero elements of $\bbF_q$ form a $(q-1)$-element subgroup of this multiplicative group of $\bbF_{q^{j+1}}$, these functionals are distinct, even modulo scalar multiplication, provided we restrict the exponent $r$ of $\gamma$ to the set $\set{0,\dotsc,R-1}$.
As such, the $R$ distinct hyperplanes of $\bbF_{q^{j+1}}$ can be written as the null spaces of the mappings $v\mapsto\tr_{q^j/q}(\gamma^{-r}v)$, that is, as $\set{\gamma^r\calS}_{r\in\calR}$ where $\calS$ is the canonical hyperplane~\eqref{equation.McFarland vs affine 0}.

These hyperplanes in hand, we are ready to construct a McFarland difference set $\calD$ in $\calG\times\calV$, where $\calG$ is any
Abelian group of order $R+1$ and $\calV$ is the additive group of $\bbF_{q^{j+1}}$.
To be precise, letting $\set{g_r}_{r=0}^{R}$ be any enumeration of $\calG$ and letting $\calR:=\set{0,\dotsc,R-1}$, McFarland~\cite{McFarland73} showed that
\begin{equation}
\label{equation.McFarland vs affine 1}
\calD
=\set{(g,v): \exists\,r\in\calR\text{ such that }g=g_r,\, v\in\gamma^r\calS}
\end{equation}
is a difference set in $\calG\times\calV$.
Moreover, as discussed in~\cite{DingF07}, these difference sets, like all Abelian difference sets, yield harmonic ETFs.
Our goal is to show that these McFarland harmonic ETFs can also be constructed by applying Theorem~\ref{theorem.Kirkman construction} in the special case where the underlying resolvable Steiner system is an affine geometry.
To accomplish this, we next find explicit expressions for the harmonic ETF that arises from~\eqref{equation.McFarland vs affine 1}.

To be precise, \cite{DingF07} gives that the restrictions of the characters of $\calG\times\calV$ to $\calD$, suitably normalized, form an ETF for $\bbC^{\calD}$.
Here, a \textit{character} of a finite Abelian group is a homomorphism from that group into the unit circle in the complex plane.
It is well known that the characters of any finite Abelian group form a unimodular orthogonal basis over that group, and moreover, that the characters of the direct product $\calG\times\calV$ are simply the tensor products of the characters of $\calG$ with those of the additive group of $\bbF_{q^{j+1}}$.

We denote the characters of $(R+1)$-element group $\calG$ as $\set{\bfchi_u}_{u=0}^{R}$.
To form the characters of $\calV$, usually called the \textit{additive characters} of $\bbF_{q^{j+1}}$, recall that $\bbF_{q^{j+1}}$ is an extension of $\bbF_q$, which in turn is an extension of its base field $\bbF_p=\langle1\rangle\cong\bbZ_p$, where the prime $p$ is the \textit{characteristic} of $\bbF_q$.
As such, we have another trace function $\tr_{q^j/p}:\bbF_{q^{j+1}}\rightarrow\bbF_p$.
Moreover, it is well known that the characters $\set{\bfe_v}_{v\in\calV}$ of $\calV$ can be formulated in terms of this trace: $\bfe_v(v'):=\exp(\frac{2\pi\rmi}P\tr_{q^j/p}(vv'))$ for all $v'\in\calV$.
Overall, for any $u=0,\dotsc,R+1$ and $v\in\calV$, we see that the $(u,v)$th character of $\calG\times\calV$ is the function $(g,v')\mapsto\bfchi_u(g)\bfe_v(v')$.

To form an ETF with the approach of~\cite{DingF07}, we restrict the domain of these characters to the difference set~\eqref{equation.McFarland vs affine 1}, and then normalize the resulting functions.
Note that since $\calD$ is parametrized in terms of $\calR$ and $\calS$ with $(g,v)=(g_r,\gamma^r s)$, we regard these restricted characters as functions over $\calR\times\calS$; for any $u=0,\dotsc,R+1$ and $v\in\calV$, consider $\bfpsi_{u,v}\in\bbC^{\calR\times\calS}$,
\begin{equation}
\label{equation.McFarland vs affine 2}
\bfpsi_{u,v}(r,s)
:=D^{-\frac12}\bfchi_u(g_r)\exp(\tfrac{2\pi\rmi}P\tr_{q^j/p}(v\gamma^r s)).
\end{equation}
Note that these ETFs have the exact dimensions~\eqref{equation.affine Kirkman parameters} of Kirkman ETFs generated via affine geometries:
they consist of $N=\abs{\calG\times\calV}=V(R+1)=q^{j+1}[(q^{j+1}-1)/(q-1)+1]$ vectors in a space of dimension $M=\abs{\calR\times\calS}=SR=q^j(q^{j+1}-1)/(q-1)$.
Noting the similarity between the formula for these restricted characters~\eqref{equation.McFarland vs affine 2} and the formula~\eqref{equation.definition of Kirkman ETF} of the Kirkman ETFs from Theorem~\ref{theorem.Kirkman construction},
it becomes even more reasonable that the two types of ETFs are, in fact, the same.

To formally make this identification, we construct these same vectors $\set{\bfpsi_{u,v}}$ via the approach of Theorem~\ref{theorem.Kirkman construction}.
In particular, we show these vectors arise from an affine geometry over the finite field $\bbF_q$ which, as stated in the previous section has $K=q$ and $V=q^{j+1}$.
Indeed, note that the vectors of any Kirkman ETF generated via such a system lie in a space of dimension $B=VR/K=q^j(q^{j+1}-1)/(q-1)=D$,
and so the normalization factors in both~\eqref{equation.McFarland vs affine 2} and~\eqref{equation.definition of Kirkman ETF} are identical.

Moreover, the fact that $\set{f_u}_{u=0}^R$, $f_u(r):=\chi_u(g_r)$ forms a unimodular simplex in $\bbC^\calR$ follows from the fact that the characters $\set{\chi_u}_{u=0}^R$ form a unimodular orthogonal basis in $\bbC^\calG$ where $\calG=\set{g_r}_{r=0}^R=\set{g_r}_{r\in\calR}\cup\set{g_R}$.
Indeed, for any $u$ and $r$ we have $\abs{f_u(r)}=\abs{\chi_u(g_r)}=1$ and for any $u\neq u'$,
\begin{align*}
\ip{f_{u}}{f_{u'}}
&=\sum_{r\in\calR}f_u(r)[f_{u'}(r)]^*\\
&=\sum_{g\in\calG}\chi_u(g)[\chi_{u'}(g)]^*-\chi_u(g_R)[\chi_{u'}(g_R)]^*\\
&=-\chi_u(g_R)[\chi_{u'}(g_R)]^*,
\end{align*}
and so $\abs{\ip{f_{u}}{f_{u'}}}=\abs{\chi_u(g_R)}\abs{\chi_{u'}(g_R)}=1$, as needed.

As such, in order to show that harmonic ETFs constructed via McFarland difference sets~\eqref{equation.McFarland vs affine 2} can be constructed as Kirkman ETFs~\eqref{equation.definition of Kirkman ETF}, we need to show that the restricted additive characters $\exp(\tfrac{2\pi\rmi}P\tr_{q^j/p}(v\gamma^r s))$ can be written in the form $\bfh_{s(r,v)}(s)$.
Here, the key observation is that any element of $\calV=\bbF_{q^{j+1}}$ can be decomposed in terms of the canonical hyperplane~\eqref{equation.McFarland vs affine 0}.
Indeed, since $\tr_{q^j/q}$ is a nontrivial linear functional of the vector space $\bbF_{q^{j+1}}$ with respect to the field $\bbF_q$, there exists $\delta\in\bbF_{q^{j+1}}$ such that $\tr_{q^j/q}(\delta)=1$.
Moreover, since $\delta$ lies outside of the hyperplane $\calS$, we can tack $\delta$ onto a basis for $\calS$ to form a basis for $\bbF_{q^{j+1}}$.
As such, any element of $\calV$ can be uniquely written as $s+t\delta$ where $s\in\calS$ and $t\in\bbF_q$.
In particular, for any $r\in\calR$ and $v\in\calV$, there exists a unique $s(r,v)\in\calS$ and $t(r,v)\in\bbF_q$ such that
\begin{equation}
\label{equation.McFarland vs affine 3}
v\gamma^r\delta
=s(r,v)+t(r,v)\delta.
\end{equation}
Note that for any $s\in\calS$, multiplying~\eqref{equation.McFarland vs affine 3} by $s\delta^{-1}$ and then applying the linear functional $\tr_{q^j/p}$ yields:
\begin{align}
\notag
\tr_{q^j/p}(v\gamma^r s)
&=\tr_{q^j/p}(s(r,v)s\delta^{-1}+t(r,v)s)\\
\label{equation.McFarland vs affine 4}
&=\tr_{q^j/p}(s(r,v)s\delta^{-1})+t(r,v)\tr_{q^j/p}(s).
\end{align}
At this point, we introduce a third trace $\tr_{q/p}:\bbF_q\rightarrow\bbF_p$ which complements the other two.
Recalling $\bbF_p\subseteq\bbF_q\subseteq\bbF_{q^j}$,
it is well known that these three traces satisfy the nice property that $\tr_{q^j/p}=\tr_{q/p}\circ\tr_{q^j/q}$.
In particular, recalling~\eqref{equation.McFarland vs affine 0} and the linearity of the trace, any $s\in\calS$ satisfies
\begin{equation}
\label{equation.McFarland vs affine 5}
\tr_{q^j/p}(s)
=\tr_{q/p}(\tr_{q^j/q}(s))
=\tr_{q/p}(0)
=0.
\end{equation}
As such,~\eqref{equation.McFarland vs affine 4} reduces to $\tr_{q^j/p}(v\gamma^r s)=\tr_{q^j/p}(s(r,v)s\delta^{-1})$,
implying that the formula~\eqref{equation.McFarland vs affine 2} for our McFarland ETF can be rewritten as
\begin{equation*}
\bfpsi_{u,v}(r,s)
=D^{-\frac12}\bfchi_u(g_r)\exp(\tfrac{2\pi\rmi}P\tr_{q^j/p}(s(r,v)s\delta^{-1})).
\end{equation*}
Comparing this to~\eqref{equation.definition of Kirkman ETF},
showing that this McFarland harmonic ETF is a Kirkman ETF boils down to showing two claims: (i) that $\set{\bfh_{s'}}_{s'\in\calS}$, $\bfh_{s'}(s):=\exp(\tfrac{2\pi\rmi}P\tr_{q^j/p}(s's\delta^{-1}))$ is a unimodular orthogonal basis for $\bbC^\calS$ and (ii) that the means of identifying $s(r,v)$ from a given $r$ and $v$ according to~\eqref{equation.McFarland vs affine 3} corresponds to a resovable Steiner system.

The truth of the first claim arises from the orthogonality of the additive characters of $\bbF_{q^{j+1}}$.
Indeed, for any $s'\neq s''$, the fact that $(s'-s'')\delta^{-1}\neq0$ gives that
\begin{equation}
\label{equation.McFarland vs affine 7}
0
=\sum_{v\in\calV}\exp(\tfrac{2\pi\rmi}P\tr_{q^j/p}((s'-s'')v\delta^{-1})).
\end{equation}
Decomposing any $v\in\calV$ as $v=s+t\delta$ where $s\in\calS$, $t\in\bbF_q$ and then using~\eqref{equation.McFarland vs affine 5} in the case where ``$s$" is $s'-s''\in\calS$  gives
\begin{equation*}
\tr_{q^j/p}((s'-s'')v\delta^{-1})
=\tr_{q^j/p}((s'-s'')s\delta^{-1}).
\end{equation*}
Substituting this into~\eqref{equation.McFarland vs affine 7} then gives our first claim:
\begin{align*}
0
&=\sum_{s\in\calS}\sum_{t\in\bbF_q}\exp(\tfrac{2\pi\rmi}P\tr_{q^j/p}((s'-s'')s\delta^{-1})\\
&=q\ip{\bfh_{s'}}{\bfh_{s''}}.
\end{align*}

For the second claim, we let~\eqref{equation.McFarland vs affine 3} \textit{define} a block design.
To be precise, let $\calB=\set{b_{r,s}}_{(r,s)\in\calR\times\calS}$ be a set of subsets of $\calV=\bbF_{q^{j+1}}$ where for any $r=0,\dotsc,R-1$ and $s\in\calS$ we say that $v\in b_{r,s}$ if and only if there exists a $t\in\bbF_q$ such that $v\gamma^r\delta=s+t\delta$; solving for $v$ reveals the $(r,s)$ block to be
\begin{equation}
\label{equation.McFarland vs affine 8}
b_{r,s}
=\set{s\gamma^{-r}\delta^{-1}+t\gamma^{-r}: t\in\bbF_q}.
\end{equation}
Recalling that for any element of $\bbF_{q^{j+1}}$ there exists exactly one $s\in\calS$ and $t\in\bbF_q$ so that it can be written as $s+t\delta$, we see that for any fixed $r\in\calR$, there exists exactly one $s=s(r,v)$ such that $v\in b_{r,s}$.
As such, every $v\in\calV$ is contained in exactly $R$ blocks and moreover, for any fixed $r\in\calR$, $\calB_r=\set{b_{r,s}}_{s\in\calS}$ forms a partition of $\calV$.
Also, every block $b_{r,s}$ contains the same number of points, namely the $K=q$ points that arise from the various choices of $t$.

Thus, in order to see that $\calB$ is a resolvable $(2,K,V)$-Steiner system over $\calV$, all that remains to be shown is that any two distinct $v,v'\in\calV$ determine a unique block.
This gets to the heart of an affine geometry over a finite field:
the blocks are affine lines~\eqref{equation.McFarland vs affine 8}, which are determined by a nonzero \textit{direction vector} $\gamma^{-r}$, which is only unique up to nonzero scalar multiples, along with an \textit{initial point} $s\gamma^{-r}\delta^{-1}$ which lies in some hyperplane.
Any two distinct points $v,v'$ determine a direction $v'-v\neq0$;
since $\set{\gamma^r}_{r=0}^{R-1}$ represent every nonzero element of $\bbF_{q^{j+1}}$ modulo scalar multiplication, we know there exists a unique $r_0\in\calR$ and $t_0\in\bbF_q$ such that $v'-v=t_0\gamma^{-r_0}$.
Moreover, for this particular $r$, we know there exists unique $s,s'\in\calS$ and $t,t'\in\bbF_q$ such that $v=s\gamma^{-r_0}\delta^{-1}+t\gamma^{-r_0}$ and $v'=s'\gamma^{-r_0}\delta^{-1}+t'\gamma^{-r_0}$, respectively.
Combining these facts gives
\begin{align*}
v'
&=(v'-v)+v\\
&=t_0\gamma^{-r_0}+(s\gamma^{-r_0}\delta^{-1}+t\gamma^{-r_0})\\
&=s\gamma^{-r_0}\delta^{-1}+(t+t_0)\gamma^{-r_0},
\end{align*}
at which point the uniqueness gives $s'=s$ and $t'=t+t_0$.
Since $s'=s$, these two points $v$ and $v'$ are contained in the same block $b_{r,s}=b_{r,s'}$.
Also, this common block is unique: if $v,v'\in b_{r,s}$ are distinct, then $v'-v$ uniquely determines $r$; knowing $v$ and $r$, $s$ is always uniquely determined.
We summarize these results in the following theorem, which is the second main result of this paper.

\begin{theorem}
\label{theorem.McFarland is affine}
Let $j\geq 1$ and let $q$ be a power of a prime $p$.
Let $\calR=\set{0,\dotsc,R-1}$ where $R=(q^{j+1}-1)/(q-1)$, and let $\calS$ be the hyperplane~\eqref{equation.McFarland vs affine 0}.
Let $\set{\chi_u}_{u=0}^{R}$ be the characters of an Abelian group $\calG=\set{g_r}_{r=0}^{R}$ and let $\gamma$ be a primitive element of $\bbF_{q^{j+1}}$, whose additive group is denoted $\calV$.

Then the harmonic ETF generated by the McFarland difference set~\eqref{equation.McFarland vs affine 1},
namely the vectors $\set{\bfpsi_{u,v}}_{u=0,v\in\calV}^{R-1}\subseteq\bbC^{\calR\times\calS}$ given in~\eqref{equation.McFarland vs affine 2},
is an example of a Kirkman ETF constructed by Theorem~\ref{theorem.Kirkman construction}.
To be precise, taking any $\delta\in\bbF_{q^{j+1}}$ such that $\tr_{q^j/q}(\delta)=1$, the blocks $\calB=\set{b_{r,s}}_{r=0,s\in\calS}^{R-1}$ defined in~\eqref{equation.McFarland vs affine 8} are a resolvable $(2,q,q^{j+1})$-Steiner system (affine geometry) which generates this same ETF, provided we let $\bff_u(r)=\chi_u(g_r)$ and $\bfh_{s'}(s):=\exp(\tfrac{2\pi\rmi}P\tr_{q^j/p}(s's\delta^{-1}))$.
\end{theorem}

We emphasize that this result tells us nothing new about the existence of ETFs.
Rather, the significance of Theorem~\ref{theorem.McFarland is affine} is that it provides (sparse) Steiner ETF representations for one of the only two known flexible classes of harmonic ETFs; as we now describe, this has ramifications on the use of such ETFs for compressed sensing.

\subsection{Kirkman ETFs and the Restricted Isometry Property}

Given $L\leq M\leq N$ and $\delta<1$, the vectors $\set{\bfphi_n}_{n\in\calN}$ in $\bbC^\calM$ have the $(L,\delta)$-\textit{Restricted Isometry Property} (RIP) if for any $L$-element subset $\calL$ of $\calN$, the eigenvalues of the Gram matrix of $\set{\bfphi_n}_{n\in\calL}$  lie in $[1-\delta,1+\delta]$.
That is, for all such $\calL$, we want $\norm{\bfPhi_\calL^*\bfPhi_\calL^{}-\bfI}_2\leq\delta$ where $\bfPhi_\calL \bfy:=\sum_{n\in\calL}\bfy(n)\bfphi_n$ is the corresponding restricted synthesis operator.
In essence, an $(L,\delta)$-RIP matrix $\bfPhi$ has the property that any $L$-element subset of its columns are nearly orthonormal.

Though other paradigms exist, this property is undeniably central to compressed sensing.
For a given $M$ and $N$, the goal is to design $\set{\bfphi_n}_{n\in\calN}$ so that it is $(L,\delta)$-RIP for $L$ being as large as possible, subject to the constraint that $\delta$ is sufficiently small compared to $1$.
To date, the most successful examples of such matrices are given by random matrix theory;
such random constructions typically yield matrices $\bfPhi$ that, with high probability, are $(L,\delta)$-RIP for $L$ on the order of $M/\operatorname{polylog}(N)$.
This is in stark contrast to nearly all deterministic constructions of such matrices which, with the exception of~\cite{Bourgain11}, are only provably $(L,\delta)$-RIP for $L$ on the order of $M^\frac12$.
In the compressed sensing literature, this is known as the \textit{square-root bottleneck}.
These facts are common knowledge, and are more thoroughly explained in~\cite{BandeiraFMW13}.

For most deterministic constructions, it is unknown whether this bottleneck is due to a lack of good proof techniques or more seriously, is due to a fault in the construction itself.
To be precise, the Gershgorin Circle Theorem gives
\begin{equation}
\label{equation.compressed sensing 1}
\norm{\bfPhi_\calL^*\bfPhi_\calL^{}-\bfI}_2
\leq\max_{n\in\calL}\sum_{\substack{n'\in\calL\\n'\neq n}}\abs{\ip{\bfphi_n}{\bfphi_{n'}}}
\leq(L-1)\mu,
\end{equation}
where $\mu$ is the coherence~\eqref{equation.definition of coherence} of $\set{\bfphi_n}_{n\in\calN}$.
To use this fact to prove that $\set{\bfphi_n}_{n\in\calN}$ is $(L,\delta)$-RIP, we thus want to choose $L$ such that $(L-1)\mu\leq\delta<1$.
In the case where $\set{\bfphi_n}_{n\in\calN}$ is a sequence of unit vectors with redundancy $\rho:=N/M$, the Welch bound~\eqref{equation.Welch bound} then yields the bottleneck:
\begin{equation}
\label{equation.compressed sensing 2}
L-1
\leq\tfrac{\delta}{\mu}
<\bigparen{\tfrac{M(N-1)}{N-M}}^\frac12
=\bigparen{\tfrac{\rho M-1}{\rho-1}}^\frac12
\leq\bigparen{\tfrac{\rho}{\rho-1}}^\frac12 M^\frac12.
\end{equation}
As such, in order to push beyond this bottleneck, we need to first find vectors for which the bounds in~\eqref{equation.compressed sensing 1} are too coarse, and then find a better way for estimating the eigenvalues of the resulting submatrices.
These are hard problems since the Gershgorin Circle Theorem, though easily proven, yields bounds which are surprisingly sharp.

Indeed, the bounds in~\eqref{equation.compressed sensing 1} are good in the case where $\set{\bfphi_n}_{n\in\calN}$ is a Steiner ETF.
To see this, recall from Section~\ref{section.Kirkman ETFs} that the Welch bound of any such ETF is $1/R$ and so~\eqref{equation.compressed sensing 2} becomes $L-1<R$.
That is, for any $L\leq R$, there exists $\delta<1$ such that $\set{\bfphi_n}_{n\in\calN}$ is $(L,\delta)$-RIP.
Remarkably, the converse of this fact is also true.
To elaborate, note that if $\set{\bfphi_n}_{n\in\calN}$ is $(L,\delta)$-RIP for some fixed $L\leq M$ and $\delta<1$, then at the very least, any $L$-element subset of $\set{\bfphi_n}_{n\in\calN}$ is linearly independent.
This means its \textit{spark}---the number of vectors in its smallest linearly dependent subcollection---is at least $L+1$.
However, as noted in~\cite{FickusMT12}, the spark of any Steiner ETF is at most $R+1$.
Indeed, in the special case where the underlying Steiner system is resolvable, note that for any fixed $v\in\calV$, the subcollection  $\set{\bfphi_{u,v}}_{u=0}^{R}$ of the Steiner ETF~\eqref{equation.definition of resolvable Steiner ETF} defined in Theorem~\ref{theorem.Kirkman construction} is only supported over the indices $(r,s)$ of those blocks $b_{r,s}$ which contain $v$.
Since there are $R+1$ such vectors but only $R$ such blocks, these vectors are necessarily linearly dependent.
As such, if $\set{\bfphi_n}_{n\in\calN}$ is $(L,\delta)$-RIP for some $\delta<1$ then $L+1\leq\operatorname{spark}(\set{\bfphi_n}_{n\in\calN})\leq R+1$.
In summary, a Steiner ETF is $(L,\delta)$-RIP for some $\delta<1$ if and only if $L\leq R$.

We now combine this fact with the main results of this section and the previous one to prove, for the first time, that some harmonic ETFs are not good RIP matrices.
Indeed, Theorem~\ref{theorem.McFarland is affine} states that every ETF generated by a McFarland difference set---one of only two known flexible constructions of harmonic ETFs---is, in fact, a Kirkman ETF.
Moreover, Theorem~\ref{theorem.Kirkman construction} states that any Kirkman ETF can be obtained by applying a unitary transformation to a Steiner ETF; it is well known that such transforms preserve RIP.
Together, we have:
\begin{corollary}
\label{corollary.Steiner RIP}
If $\set{\bfphi_n}_{n\in\calN}$ is any Steiner or Kirkman ETF for $\bbC^\calM$, then it is $(L,\delta)$-RIP for some $\delta<1$ if and only if
\begin{equation*}
L\leq R=\bigparen{\tfrac{\rho M-1}{\rho-1}}^\frac12,
\end{equation*}
where $\rho=\frac NM$.
In particular, it is impossible to surpass the square-root bottleneck using harmonic ETFs generated from McFarland difference sets.
\end{corollary}

It remains an open problem whether or not there exists an ETF which is a good RIP matrix for values of $L$ which are larger than numbers on the order of $M^\frac12$.
However, in light of Corollary~\ref{corollary.Steiner RIP}, there is only one known flexible class of ETFs left to investigate, namely the harmonic ETFs generated by Singer difference sets.
These difference sets are cyclic, meaning these ETFs are obtained by extracting $M=(q^j-1)/(q-1)$ rows from a standard DFT matrix of size $N=(q^{j+1}-1)/(q-1)$ where $j\geq 2$ and $q$ is some prime power.
Here, there are some reasons for hope: when $j=2$, such ETFs are \textit{numerically erasure-robust frames} and as such, cannot be sparse in any basis~\cite{FickusM12}.
Moreover, Singer harmonic ETFs are \textit{full spark}---their spark is $M+1$---when $N$ is prime, such as when $j=2$ and $q=2$.
But even this can fail when $N$ is but a prime power, such as when $j=4$ and $q=3$~\cite{AlexeevCM12}.

Also, there are inflexible families of non-Steiner ETFs whose RIP characteristics bear further study.
One example of these are  \textit{Paley ETFs} which are constructed by modifying a quadratic-residue-based harmonic ETF into a redundancy-two ETF in the manner of~\cite{Renes07}.
There at least, spark is not the issue: any Paley ETF is $(L,\delta)$-RIP for all $L\leq M$ for some $\delta<1$~\cite{BandeiraFMW13}.
However, it is unknown how this $\delta$ behaves as a function of $L$, $M$ and $N$;
this is related to longstanding open problems regarding the clique numbers of Paley graphs~\cite{BandeiraFMW13}.
These problems are nontrivial, and it is much easier to prove that a given set of vectors is not $(L,\delta)$-RIP than to prove that it is.
Put simply, Corollary~\ref{corollary.Steiner RIP} does not tell you where to find good RIP matrices but rather, where not to look.

\section{Hadamard ETFs and the Grey-Rankin Bound}
\label{section.Grey-Rankin}

In this section, we apply the results of Sections~\ref{section.Kirkman ETFs} and~\ref{section.Kirkman vs harmonic} to produce new examples of certain types of optimal binary codes.
An \textit{$(M,N)$-binary code} is a set of $N$ codewords (vectors) in $\bbZ_2^M$,
that is, a sequence $\set{\bfc_n}_{n=1}^{N}$ of $M\times 1$ vectors whose entries lie in $\bbZ_2:=\set{0,1}$.
The \textit{distance} of such a code is the minimum pairwise Hamming distance between any two codewords,
namely $\dist(\set{\bfc_n}_{n=1}^{N}):=\min_{n\neq n'}\hd(\bfc_n,\bfc_{n'})$ where $\hd(\bfc,\bfc')$ counts the number of entries of $\bfc,\bfc'\in\bbZ_2^M$ that differ.
The \textit{Grey-Rankin bound}~\cite{Grey62} is an upper bound on the number of codewords $N$ one can have with a given distance $\Delta$ in a space with given dimension $M$.
We show that the Grey-Rankin bound is equivalent to a special case of the Welch bound, and then exploit this equivalence, using coding theory to prove new results in frame theory, and vice versa.

To be precise, the Grey-Rankin bound only applies to \textit{self-complementary codes}, that is, codes in which the \textit{complement} $(\bfc_n+\bfone)(m):=\bfc_n(m)+1\bmod2$ of any codeword $\bfc_n$ also lies in the code.
In the work that follows, it is convenient for us to regard the second half of these $(M,2N)$-binary codes as the complements of the first half, namely $\bfc_{n+N}=\bfc_n+\bfone$ for all $n=1,\dotsc,N$.
Denoting the distance of such a code as $\Delta$, the Grey-Rankin bound states:
\begin{equation}
\label{equation.Grey-Rankin}
2N\leq\tfrac{8\Delta(M-\Delta)}{M-(M-2\Delta)^2}
\end{equation}
provided the right-hand side is positive; since the self-complementarity of the code guarantees that $2\Delta\leq M$, this positivity is equivalent to having $2\Delta>M-M^\frac12$.

We rederive~\eqref{equation.Grey-Rankin} by applying the Welch bound~\eqref{equation.Welch bound} to frames whose entries are all $\pm M^{-\frac12}$.
To be clear, we can exponentiate any codeword $\bfc_n\in\bbZ_2^M$ to form a corresponding unit norm vector $\bfphi_n\in\bbR^M$, $\bfphi_n(m):=M^{-\frac12}(-1)^{\bfc_n(m)}$.
Under this identification, the Euclidean distance between any two of these real vectors can be written in terms of the Hamming distance between their corresponding codewords:
\begin{align*}
\norm{\bfphi_n-\bfphi_{n'}}^2
&=\tfrac1M\sum_{m=1}^{M}\abs{(-1)^{\bfc_n(m)}-(-1)^{\bfc_{n'}(m)}}^2\\
&=\tfrac1M\sum_{m=1}^{M}\left\{\begin{array}{cc}4,&\bfc_n(m)\neq \bfc_{n'}(m)\\0,&\bfc_n(m)=\bfc_{n'}(m)\end{array}\right\}\\
&=\tfrac4M\hd(\bfc_n,\bfc_{n'}).
\end{align*}
We also have $\norm{\bfphi_n-\bfphi_{n'}}^2=2(1-\ip{\bfphi_n}{\bfphi_{n'}})$, and solving for the inner product gives
\begin{equation}
\label{equation.inner products in term of Hamming}
\ip{\bfphi_n}{\bfphi_{n'}}
=\tfrac1{M}(M-2\hd(\bfc_n,\bfc_{n'})).
\end{equation}
Grey himself used this identification in his derivation of~\eqref{equation.Grey-Rankin}.
However, Grey's argument~\cite{Grey62} relies on prior work by Rankin~\cite{Rankin56} concerning the packing of spherical caps,
whereas we instead make use of Welch's bound~\eqref{equation.Welch bound}.

The mathematical novelty here is debatable: Rankin's work is a forerunner to Welch's bound.
In fact, a little simplification reveals Equation~26 of~\cite{Rankin56} to be equivalent to the real-variable version of the Welch bound;
since it predates~\cite{Welch74} by nearly two decades, one can argue that~\eqref{equation.Welch bound} should be called the ``Rankin-Welch" bound.
From this perspective, our work below serves to modernize Grey's original argument.
This itself has value: unlike Rankin's work, the Welch bound is widely studied.
Also, as seen from~\eqref{equation.derivation of Welch bound}, the Welch bound can be quickly proven from basic principles.
Most importantly, by using the Welch bound to streamline Grey's approach,  we allow the large body of existing ETF/WBE literature to be quickly and directly applied to open problems in coding theory.

Returning to the argument itself, taking the maximums of both sides of~\eqref{equation.inner products in term of Hamming} over all $n,n'=1,\dotsc, 2N$, $n\neq n'$ gives
\begin{equation}
\label{equation.Grey-Rankin derivation 1}
\max_{\substack{n,n'\in\set{1,\dotsc,2N}\\n\neq n'}}\ip{\bfphi_n}{\bfphi_{n'}}
=\tfrac{M-2\Delta}{M}.
\end{equation}
Moreover, the self-complementarity of the code $\set{\bfc_n}_{n=1}^{2N}$ gives that $\bfphi_{n+N}=-\bfphi_n$ for all $n=1,\dotsc,N$.
As such, the left-hand side of~\eqref{equation.Grey-Rankin derivation 1} can be rewritten as
\begin{equation}
\label{equation.Grey-Rankin derivation 2}
\mu
=\max_{\substack{n,n'\in\set{1,\dotsc,N}\\n\neq n'}}\abs{\ip{\bfphi_n}{\bfphi_{n'}}}
=\tfrac{M-2\Delta}{M}.
\end{equation}
At this point, the Welch bound~\eqref{equation.Welch bound} gives
\begin{equation}
\label{equation.Grey-Rankin derivation 3}
\bigparen{\tfrac{N-M}{M(N-1)}}^\frac12
\leq \tfrac{M-2\Delta}{M}.
\end{equation}
Squaring both sides and then solving for $N$ then gives the Grey-Rankin bound~\eqref{equation.Grey-Rankin}.
Moreover, note that by this argument, we obtain equality in~\eqref{equation.Grey-Rankin} if and only if we have equality in \eqref{equation.Grey-Rankin derivation 3}, which in light of~\eqref{equation.Grey-Rankin derivation 2}, happens precisely when $\set{\bfphi_n}_{n=1}^{N}$ is a real-valued constant-amplitude ETF.
That is, every \textit{Grey-Rankin-bound-equality} (GRBE) code generates such an ETF.

Importantly, the converse is also true:
if $\set{\bfphi_n}_{n=1}^{N}$ is any constant-amplitude ETF for $\bbR^M$,
then we may build codewords $\set{\bfc_n}_{n=1}^{N}$ by letting $\bfc_n(m):=\log_{-1}(\sgn(\bfphi_n(m))\in\bbZ_2$ for all $m=1,\dotsc,M$ and all $n=1,\dotsc,N$.
We then extend this to a self-complementary code by letting $\bfc_{n+N}:=\bfc_n+\bfone$ for all $n=1,\dotsc,N$.
Since exponentiating these codewords produces our original ETF, we know that~\eqref{equation.Grey-Rankin derivation 3} holds with equality, meaning we also have equality in the Grey-Rankin bound~\eqref{equation.Grey-Rankin}.
For example, the real $6\times 16$ Kirkman ETF given in Figure~\ref{figure.6x16 Psi} yields the $6\times 32$ self-complementary code given in Figure~\ref{figure.6x32 C}.
\begin{figure*}
\begin{align*}
\left[\begin{array}{cccccccccccccccccccccccccccccccccccc}
0&1&0&1&0&1&0&1&0&1&0&1&0&1&0&1&1&0&1&0&1&0&1&0&1&0&1&0&1&0&1&0\\
0&1&0&1&0&1&0&1&1&0&1&0&1&0&1&0&1&0&1&0&1&0&1&0&0&1&0&1&0&1&0&1\\
0&0&1&1&0&0&1&1&0&0&1&1&0&0&1&1&1&1&0&0&1&1&0&0&1&1&0&0&1&1&0&0\\
0&0&1&1&1&1&0&0&0&0&1&1&1&1&0&0&1&1&0&0&0&0&1&1&1&1&0&0&0&0&1&1\\
0&1&1&0&0&1&1&0&0&1&1&0&0&1&1&0&1&0&0&1&1&0&0&1&1&0&0&1&1&0&0&1\\
0&1&1&0&1&0&0&1&1&0&0&1&0&1&1&0&1&0&0&1&0&1&1&0&0&1&1&0&1&0&0&1
\end{array}\right]
\end{align*}
\caption{\label{figure.6x32 C}
A $(6,32)$-binary code, the left half of which is obtained by converting the $+$'s and $-$'s of the ETF in Figure~\ref{figure.6x32 C} into $0$'s and $1$'s, respectively.  This code is self-complementary, meaning its right half is obtained by adding $1$ to the left half, modulo $2$.
This code achieves the Grey-Rankin bound for $M=6$ and $\Delta=2$, meaning it is the widest possible self-complementary matrix of height $6$ such that the Hamming distance of any two columns is at least $2$.
We show that such Grey-Rankin-bound-equality matrices are equivalent to real-valued constant-amplitude ETFs, and then exploit this result to prove new results about ETFs using coding theory, and vice versa.
}
\end{figure*}
We summarize these results as our third main result:
\begin{theorem}
\label{theorem.Grey-Rankin Welch bound}
Any $(M,2N)$-binary self-complementary code $\set{\bfc_n}_{n=1}^{2N}$ with $\bfc_{n+N}=\bfc_n+\bfone$ for all $n=1,\dotsc,N$ satisfies the Grey-Rankin bound~\eqref{equation.Grey-Rankin}.
Moreover, identifying $\set{\bfc_n}_{n=1}^{N}$ with constant-amplitude vectors $\set{\bfphi_n}_{n=1}^{N}\subseteq\bbR^M$ according to $\bfphi_n(m)=M^{-\frac12}(-1)^{\bfc_n(m)}$, the code $\set{\bfc_n}_{n=1}^{2N}$ achieves the Grey-Rankin bound~\eqref{equation.Grey-Rankin} if and only if $\set{\bfphi_n}_{n=1}^{N}$ is an ETF.
\end{theorem}

In light of Theorem~\ref{theorem.Grey-Rankin Welch bound}, we turn our attention to the problem of constructing \textit{real-valued} constant-amplitude ETFs.
As noted in~\cite{DingF07}, harmonic examples of such ETFs can be constructed over the additive group of $\bbF_{2^{2j+2}}$ by forming a difference set $\calD$ as the support of a \textit{bent function}.
Such ETFs have parameters
\begin{equation}
\label{equation.bent parameters}
M=2^j(2^{j+1}\pm1),
\quad
N=2^{2j+2}
\end{equation}
for some $j\geq1$.
In the context of the previous sections, it is easier to understand ETFs with parameters~\eqref{equation.bent parameters} as special cases of harmonic ETFs generated from McFarland difference sets.
Indeed, $M=2^j(2^{j+1}-1)$ and $N=2^{2j+2}$  is a special case of~\eqref{equation.affine Kirkman parameters} where $q=2$.
Here, the corresponding McFarland difference set~\eqref{equation.McFarland vs affine 1} lies in the group $\calH=\calG\times\calV$ where $\calV$ is the additive group of $\bbF_{2^{j+1}}$ and $\calG=\bbZ_2^{j+1}$.
The resulting ETF is real since every element of $\calH$ has order $2$:
for any $h\in\calH$ and any character $\bfchi$ of $\calH$, $[\bfchi(h)]^2=\bfchi(h+h)=\bfchi(0)=1$ and so $\bfchi(h)=\pm1$.
Moreover, the set complement of any difference set is another difference set~\cite{JungnickelPS07}.
Thus, there also exists a real-valued constant-amplitude ETF of $N=2^{2j+2}$ vectors in a space of dimension $N-M=2^j(2^{j+1}+1)$;
this complementary ETF is a special case of the \textit{Naimark complement} of a tight frame $\set{\bfphi}_{n\in\calN}$, which in general, is formed by finding a orthonormal basis for the orthogonal complement of the row space of $\bfPhi$.

Surveying the literature, we find that all known real-valued harmonic ETFs are either regular simplices or have parameters~\eqref{equation.bent parameters}.
As we now explain, these are the only possible dimensions for a real-valued harmonic ETF.
To be precise, a code is \textit{linear} if the codewords are the points in some subspace of $\bbZ_2^M$.
And, in the case where the underlying real ETF is harmonic, the corresponding code generated by Theorem~\ref{theorem.Grey-Rankin Welch bound} is necessarily linear: signed characters are closed under multiplication, and so their corresponding codewords are closed under addition.
Moreover, it is known that a linear GRBE code with $M\geq 2$ must either have dimensions $M=2^{j+1}-1$ and $2N=2^{j+2}$ for some $j\geq1$---meaning its corresponding ETF is a real-valued constant-amplitude regular simplex---or alternatively, dimensions $M=2^j(2^{j+1}\pm1)$ and $2N=2^{2j+3}$ for some $j\geq 1$~\cite{McGuire97}.  As such, we find that:
\begin{corollary}
\label{corollary.real harmonic}
If there exists a real-valued harmonic ETF of $N$ vectors in an $M$-dimensional space with $M\geq 2$, then either (i) the ETF is a regular simplex of $N=2^{j+1}$ vectors for some $j\geq1$ or (ii) the dimensions of the ETF are of the form~\eqref{equation.bent parameters}.
\end{corollary}

This corollary illustrates how coding theory can be used to find new results in frame theory.
However, from the point of view of coding theory itself, this corollary is disappointing: GRBE codes with these parameters are already known to exist~\cite{McGuire97}.
It is here that the not-necessarily-harmonic ETFs of Theorem~\ref{theorem.Kirkman construction} truly shine: we can find Kirkman ETFs that lie outside the confines of Corollary~\ref{corollary.real harmonic};
these ETFs are necessarily non-harmonic, and the resulting codes are necessarily nonlinear.

To be precise, in order to construct real-valued constant-amplitude ETFs using Theorem~\ref{theorem.Kirkman construction}, we want both our unimodular regular simplex $\set{\bff_{u}}_{u=0}^{R}$ as well as our unimodular orthogonal basis $\set{\bfh_{s}}_{s\in\calS}$ to be real-valued.
Since such a simplex necessarily extends to a real unimodular orthogonal basis, we in particular want Hadamard matrices of size $R+1$ and $S=B/R=V/K$.
It is well-known that this requires both $R+1$ and $V/K$ to either be $2$ or divisible by $4$; the \textit{Hadamard conjecture} posits that these necessary conditions are sufficient.
Also, recall that in order for a resolvable $(2,K,V)$-Steiner system to exist, we necessarily have $V\equiv K\bmod K(K-1)$.
Writing $V=WK(K-1)+K$ for some $W\geq1$, we want
\begin{equation}
\label{equation.Hadamard necessary 1}
R+1=WK+2,\quad \tfrac VK=W(K-1)+1
\end{equation}
to either be $2$ or divisible by $4$.
Note that since $W\geq1$ and $K\geq2$, we cannot have $R=2$, and so this condition on $R$ is equivalent to having $WK+2\equiv 0\bmod 4$.
Meanwhile, if $W(K-1)+1=2$ then we necessarily have $K=2$ and $W=1$; the resulting $(2,2,4)$-Steiner system~\eqref{equation.6x16 B} yields the $6\times 32$ code of Figure~\ref{figure.6x32 C}; as discussed above, linear GRBE codes with these parameters are well-known, and can be generated as McFarland harmonic ETFs, letting $j=1$ in~\eqref{equation.bent parameters}.
As such, we also assume $W(K-1)+1\equiv 0\bmod 4$.
At this point, subtracting $W(K-1)+1$ from $WK+2$, we see that these two necessary conditions are equivalent to having $W\equiv 3\bmod 4$ and $K\equiv 2\bmod 4$.
Combining the above discussion with Theorems~\ref{theorem.Kirkman construction} and~\ref{theorem.Grey-Rankin Welch bound} gives the following result:
\begin{corollary}
\label{corollary.real Kirkman}
Given $K\equiv 2\bmod 4$ and $W\equiv 3\bmod 4$, let $V=K[W(K-1)+1]$.
Then both parameters in~\eqref{equation.Hadamard necessary 1} are divisible by $4$, and if there exist Hadamard matrices of these sizes and there also exists a resolvable $(2,K,V)$-Steiner system, then there exists a real-valued Kirkman ETF with
\begin{align*}
M&=(WK+1)[W(K-1)+1],\\
N&=K(WK+2)[W(K-1)+1],
\end{align*}
meaning there exists a $(M,2N)$-self-complementary code that achieves the Grey-Rankin bound~\eqref{equation.Grey-Rankin}.
\end{corollary}

To explore the consequences of this result, we first consider the case where $K=2$.
Recall from Section~\ref{section.Kirkman ETFs} that $(2,2,V)$-Steiner systems are resolvable as a round-robin tournament for any even $V\geq 4$.
Here, $V=K[W(K-1)+1]=2W+2$ for some $W\equiv 3\bmod 4$.
To make the resulting Kirkman ETF real-valued, we want Hadamard matrices of size $R+1=2W+2=V$ and $V/K=W+1=V/2$.
It thus suffices for there to exist a Hadmard matrix of size $V/2$, since we can take the tensor product of it with the canonical Hadamard matrix of size $2$ to form one of size $V$.
When $V$ is a power of $2$, the ETFs produced by Corollary~\ref{corollary.real Kirkman} have the same dimensions as those produced by the real-valued harmonic ETFs of~\eqref{equation.bent parameters}.
However, when $V$ is not a power of $2$, Corollary~\ref{corollary.real harmonic} tells us that these Kirkman ETFs cannot be harmonic.
For example, letting $W=11$ yields $V=24$, and we know there exists a Hadamard matrix of size $V/2=12$.
As such, there exists a real-valued $276\times 576$ Kirkman ETF that cannot be harmonic,
and the resulting $276\times 1152$ GRBE code is not linear.
There are an infinite number of nonharmonic Kirkman ETFs of this type: at the very least, Paley's quadratic-residue based construction of Hadamard matrices gives the existence of such an ETF whenever $W\equiv 3\bmod 4$ is a prime power.

For $K\equiv 2\bmod 4$ such that $K>2$, the true implications of Corollary~\ref{corollary.real Kirkman} are harder to ascertain.
We could not find any explicit infinite families of resolvable $(2,K,V)$-Steiner systems for such values of $K$ in the literature.
For $K=6$ in particular, we need $V\equiv 6\bmod 30$; it is known~\cite{FurinoMY96} that a resolvable $(2,6,V)$-Steiner system (i) does not exist for $V=36$, (ii) may or may not exist for $V=66$ and $V=96$, (iii) does exist for $V=126$ (unital design), $V=156$ (projective geometry) and $V=186$.
Unfortunately, the only one of these values of $V$ that satisfies the hypotheses of Corollary~\ref{corollary.real Kirkman} is $V=96$.
If a resolvable $(2,6,96)$-Steiner system does exist it would, to our knowledge, give the first example of a GRBE code whose redundancy $2N/M=3840/304$ is not approximately $4$.
For $K=10$ and larger, the minimum corresponding $V$ which could satisfy the assumptions of Corollary~\ref{corollary.real Kirkman} is $V=280$, which lies beyond the range of the tables of known resolvable designs we encountered.

At this point, we turn to asymptotic existence results.
Recall that for any $K$, there exists $V_0(K)$ such that for all $V\geq V_0(K)$ with $V\equiv K\bmod K(K-1)$, there exists a resolvable $(2,K,V)$-Steiner system~\cite{RayW73}.
As such, if the Hadamard conjecture is true, then for any $K\equiv 2\bmod 4$, there exists $W_0(K)$ such that for all $W\geq W_0(K)$ with $W\equiv 3\bmod 4$, there exists a real-valued Kirkman ETF whose parameters are given by Corollary~\ref{corollary.real Kirkman}.
In particular, if the Hadamard conjecture is true, for any $K\equiv 2\bmod4$ there exists real-valued constant-amplitude ETFs and GRBE codes whose redundancies are approximately $K$ and $2K$, respectively.

\section{Conclusions and Future Work}
We now have a method for transforming certain Steiner ETFs into constant-amplitude ETFs, as desired for certain waveform design applications.
We also now know that an important class of previously discovered harmonic ETFs arise in this fashion, making them less attractive for deterministic compressed sensing.
Finally, we have seen how the problem of constructing a real-valued constant-amplitude ETF is equivalent to that of constructing a type of optimal binary code, allowing us to apply results from one area to the other.
Several important questions remain open:
To what degree do Singer harmonic ETFs satisfy RIP?
Are Denniston Kirkman ETFs harmonic?
More generally, can we use these results along with ideas from resolvable designs to build new examples of difference sets, or vice versa?
Do there exist nontrivial real-valued constant-amplitude ETFs whose redundancy is not essentially $2$?
Equivalently, do there exist nontrivial GRBE codes whose redundancy is not essentially $4$?

\section*{Acknowledgment}
This work was supported by NSF DMS 1042701 and NSF CCF 1017278.
The views expressed in this article are those of the authors and do not reflect the official policy or position of the United States Air Force, Department of Defense, or the U.S.~Government.

\end{document}